\documentclass[reqno,12pt]{amsart}

\usepackage[margin=1in]{geometry}
\usepackage{amsmath,amssymb,amsthm,mathtools}
\usepackage[dvipsnames]{xcolor}
\usepackage{enumitem}
\usepackage{tikz}
\usepackage{graphicx}
\colorlet{LinkBlue}{MidnightBlue!99!black}
\usepackage[colorlinks=true,linkcolor=LinkBlue,citecolor=LinkBlue,urlcolor=LinkBlue,hypertexnames=false]{hyperref}
\usepackage{aliascnt}
\usepackage{environ}
\usepackage[nameinlink,noabbrev]{cleveref}
\usetikzlibrary{arrows.meta,calc,patterns}
\usepackage{autonum}
\usepackage[framemethod=tikz]{mdframed}

\colorlet{stylecolor}{RoyalBlue}
\colorlet{assigncolor}{Brown}

\makeatletter
\patchcmd{\@maketitle}{\global\topskip42\p@}{\global\topskip14\p@}{}{}
\patchcmd{\@maketitle}{\dimen@34\p@}{\dimen@16\p@}{}{}
\makeatother

\newcommand{\R}{\mathbb{R}}

\newcommand{\Z}{\mathbb{Z}}
\newcommand{\supp}{\operatorname{supp}}
\newcommand{\C}{\mathbb{C}}

\newcommand{\lr}[1]{\left\langle #1 \right\rangle}

\newcommand{\dist}{\mathsf{d}}
\newcommand{\defeq}{\stackrel{\mathrm{def}}{=}}
\newcommand{\1}{\mathbf{1}}

\usepackage[normalem]{ulem}

\NewEnviron{equations}{%
\begin{equation}\begin{gathered}
  \BODY
\end{gathered}\end{equation}
}

\newtheoremstyle{chen}%
  {6pt}{6pt}%
  {\normalfont}%
  {0pt}%
  {\sffamily\bfseries\color{stylecolor}}%
  {}%
  {\newline}%
  {\thmname{#1}\thmnumber{ #2}\thmnote{: #3}}

\newtheoremstyle{chenbrown}%
  {6pt}{6pt}{\normalfont}{0pt}
  {\sffamily\bfseries\color{assigncolor}}{}{\newline}
  {\thmname{#1}\thmnumber{ #2}\thmnote{: #3}}

\theoremstyle{chen}
\newtheorem{theorem}{Theorem}[section]
\newaliascnt{proposition}{theorem}
\newtheorem{proposition}[proposition]{Proposition}
\aliascntresetthe{proposition}
\newaliascnt{lemma}{theorem}
\newtheorem{lemma}[lemma]{Lemma}
\aliascntresetthe{lemma}
\newaliascnt{corollary}{theorem}

\aliascntresetthe{corollary}
\theoremstyle{chenbrown}
\newaliascnt{definition}{theorem}
\newtheorem{definition}[definition]{Definition}
\aliascntresetthe{definition}
\theoremstyle{definition}
\newaliascnt{example}{theorem}

\aliascntresetthe{example}
\theoremstyle{chenbrown}
\newaliascnt{remark}{theorem}

\aliascntresetthe{remark}
\theoremstyle{chenbrown}
\newaliascnt{assumption}{theorem}

\aliascntresetthe{assumption}

\crefname{theorem}{theorem}{theorems}
\Crefname{theorem}{Theorem}{Theorems}
\crefname{proposition}{proposition}{propositions}
\Crefname{proposition}{Proposition}{Propositions}
\crefname{lemma}{lemma}{lemmas}
\Crefname{lemma}{Lemma}{Lemmas}
\crefname{corollary}{corollary}{corollaries}
\Crefname{corollary}{Corollary}{Corollaries}
\crefname{definition}{definition}{definitions}
\Crefname{definition}{Definition}{Definitions}
\crefname{example}{example}{examples}
\Crefname{example}{Example}{Examples}
\crefname{remark}{remark}{remarks}
\Crefname{remark}{Remark}{Remarks}
\crefname{assumption}{assumption}{assumptions}
\Crefname{assumption}{Assumption}{Assumptions}
\crefname{equation}{equation}{equations}
\Crefname{equation}{Equation}{Equations}

\mdfsetup{
  linecolor=stylecolor,
  linewidth=1pt,
  backgroundcolor=stylecolor!15!white,
  innertopmargin=6pt,
  innerbottommargin=6pt,
  skipabove=8pt,
  skipbelow=8pt,
  needspace=7\baselineskip,
}
\surroundwithmdframed{theorem}
\surroundwithmdframed{proposition}
\surroundwithmdframed{lemma}
\surroundwithmdframed{corollary}
\surroundwithmdframed[linecolor=assigncolor,backgroundcolor=assigncolor!15!white]{definition}
\surroundwithmdframed[linecolor=assigncolor,backgroundcolor=assigncolor!15!white]{remark}
\surroundwithmdframed[linecolor=assigncolor,backgroundcolor=assigncolor!15!white]{assumption}

\numberwithin{equation}{section}

\makeatletter
\def\l@subsection{\@tocline{2}{0pt}{3pc}{5pc}{}}
\makeatother

\usetikzlibrary{decorations.markings,intersections}
\graphicspath{{Pictures/}}
\newcommand{\X}{\mathcal{X}}

\newcommand{\Tr}{\operatorname{Tr}}
\newcommand{\tr}{\operatorname{tr}}
\newcommand{\sgn}{\operatorname{sgn}}
\newcommand{\diam}{\operatorname{diam}}
\newcommand{\conv}{\operatorname{conv}}
\newcommand{\bx}{\mathbf{x}}
\newcommand{\by}{\mathbf{y}}
\newcommand{\Gg}{\mathbf{G}}
\newcommand{\Ff}{\mathbf{F}}
\newcommand{\w}{\omega}
\newcommand{\p}{\partial}
\colorlet{cone}{red!70!black}
\colorlet{ctwo}{blue!70!black}
\tikzset{mk/.style={decoration={markings,
  mark=at position 0.1 with {\arrow{Stealth[length=2.2mm]}},
  mark=at position 0.5 with {\arrow{Stealth[length=2.2mm]}},
  mark=at position 0.9 with {\arrow{Stealth[length=2.2mm]}}},postaction={decorate}}}

\title{Geometric bulk conductances in even dimensions}
\author{Alexis Drouot}
\address{Department of Mathematics, University of Washington, Seattle, WA 98195, USA, and 
Centre de Physique Th\'eorique, Aix-Marseille Universit\'e, Marseille, France}
\email{adrouot@uw.edu}

\begin{document}

\begin{abstract}
Let $P$ be a short-range projector on $\ell^2(\mathbb{Z}^{2n},\mathbb{C}^m)$ and $U=(U_1,\dots,U_{2n})$ be a transverse family of open subsets of $\mathbb{R}^{2n}$. We give an elementary proof of a recent result of Martinian--Stoiber, which extends a two-dimensional result of Drouot--Zhu to all even dimensions. Specifically, using only Stokes' theorem and the Poincar\'e lemma, we show that the geometric Kubo formula $\operatorname{Tr}\big(P[P,\mathbf{1}_{U_1}]\wedge\cdots\wedge[P,\mathbf{1}_{U_{2n}}]\big)$ is the product of a topological integer $\mathcal{X}_U$, which depends only on the relative positions of $U_1, \dots, U_{2n}$, and of the bulk conductance of $P$.   Our proof writes the geometric bulk conductance as the pairing of $P$ with a cochain on $\mathbb{Z}^{2n}$ -- an idea originally due to Kitaev -- and deforms the cochain rather than the sets $(U_1, \dots, U_{2n})$. The proof presented in this paper is due to Claude Opus 5.5.
\end{abstract}

\maketitle

\section{Setting and main result}\label{sec:setting}

Let $d = 2n$ and let $P$ be a short-range projector on $\ell^2(\Z^d,\C^m)$: $P^2=P$ and for all $N \geq 0$, there exists $C_N >0$ such that the kernel $P(x,y)$ of $P$ satisfies
\begin{equation}\label{eq:P-decay}
    |P(x,y)| \le C_N\lr{x-y}^{-N}.
\end{equation}
Let $U_1, \dots, U_d$ be subsets of $\R^d$ and let $\1_{U_j}$ denote the multiplication operator on $\ell^2(\Z^d,\C^m)$ by the indicator function of $U_j$. Under a transversality condition on subsets $U_1, \dots, U_d$ of $\R^d$ (see \Cref{def:transverse}), we define the geometric Kubo formula:
\begin{equation}\label{eq:cU}
    c_U(P) \defeq \Tr\big(P\,[P,\1_{U_1}]\wedge\cdots\wedge[P,\1_{U_d}]\big).
\end{equation}

The quantity $c_U(P)$ is a topological marker of insulators in even dimensions \cite{prodan-schulzbaldes}; for background about the $d=2$ case and its relation to condensed matter physics, more specifically to the bulk-edge correspondence, see \cite{DZ24}. This note provides an elementary, AI-made proof of a recent result of Martinian--Stoiber \cite{MS26}, which states that $c_U(P)$ factors out in terms of its geometric and analytic data:
\begin{equation}\label{eq:main-intro}
    c_U(P) = \X_U \cdot c_H(P).
\end{equation}
In \eqref{eq:main-intro}, $\X_U$ is an intersection number between the boundaries $\p U_1, \dots, \p U_d$ and $c_H(P)$ is the standard $d$-dimensional Kubo formula \cite[Eq.~(2.15)]{prodan-schulzbaldes}, obtained from \eqref{eq:cU} by choosing half-spaces for the sets $U_j$: $H = (H_1,\dots,H_d)$, $H_j = \{x_j>0\}$. The notions of transversality, geometric Kubo formula \eqref{eq:cU}, intersection number, and the identity \eqref{eq:main-intro} are all higher-dimensional analogues of the case $d=2$ analyzed in \cite{DZ24}.

\subsection{Setup} We introduce a few objects and conditions needed to state our main result. First, we extend the transversality condition introduced in \cite[Def.~4]{DZ24} to higher dimensions:

\begin{definition}\label{def:transverse}
A family $U = (U_1,\dots,U_d)$ of open subsets of $\R^d$ is transverse if, for some $c\in(0,1)$,
\begin{equation}
    \Psi_U(x) \defeq 1+\sum_{j=1}^d\dist(x,\p U_j) \ \ge\ |x|^c, \qquad |x| \ge c^{-1}.
\end{equation}
\end{definition}

Transversality is a condition on the whole family: pairs of boundaries may stay close, as long as the $d$ boundaries have no ``close'' point at infinity. It allows us to make sense of the geometric Kubo formula \eqref{eq:cU}.

In order to state our main result, we need to extend to higher dimensions the ``intersection number'' introduced in \cite{DZ24}. 

\begin{definition}\label{def:X}
Fix a radial, nonnegative $\eta \in C_0^\infty\big(B(0,1)\big)$ with $\int_{\R^d} \eta = 1$, and set
\begin{equation}
    f_j \defeq \1_{U_j}*\eta, \qquad F_U \defeq (f_1,\dots,f_d) : \R^d\to[0,1]^d, \qquad J_U \defeq \det(\nabla f_1,\dots,\nabla f_d).
\end{equation}
The intersection number of $U$ is
\begin{equation}\label{eq:X}
    \X_U \defeq \int_{\R^d}J_U(x)\,dx.
\end{equation}
\end{definition}

This definition was suggested to us by Guo Chuan Thiang.   Here are some quick comments:
\begin{itemize}
    \item $J_U$ is compactly supported, see \Cref{fig:support}; hence $\X_U$ is well-defined. Indeed, if $\dist(x,\p U_j) \ge 1$, then $f_j$ is constant ($0$ or $1$) near $x$. So $J_U(x)\neq0$ forces $\Psi_U(x)\le 1+d$, and hence by transversality, $x$ must be in a compact set. 
    \item For $R$ large, $F_U(S_R)\subset\p[0,1]^d$, hence $\X_U$ is the degree of $F_U : S_R\to\p[0,1]^d$. In particular, it is an integer and, from a continuity argument, independent of $\eta$.
\end{itemize}
We relate \eqref{eq:X} to the original \cite[Def.~9]{DZ24}, corresponding to $d=2$:
\begin{itemize}
\item If the boundaries $\p U_j$ are smooth and meet transversally, then, with $\nu_j$ the unit normal to $\p U_j$ pointing into $U_j$,
    \begin{equation}\label{eq:1a}
        \X_U = \sum_{p\in\p U_1\cap\dots\cap\p U_d}\sgn\det\big(\nu_1(p),\dots,\nu_d(p)\big).
    \end{equation}
    In particular $\X_U$ extends the intersection number introduced in \cite{DZ24}; and $\X_H = 1$.
    \item If the $\p U_j$ are smooth and meet transversally, and $\gamma_j$ are curves representing the connected components of $\p U_2\cap\dots\cap\p U_d$ oriented by their tangent $\tau_j$, such that
    \begin{equation}
        \det(\tau_j,\nu_2,\dots,\nu_d)>0,
    \end{equation}
    then $\X_U$ is the algebraic number of times that the curves $\gamma_j$ enter $U_1$. In other words, $\X_U$ can be interpreted in terms of topological transport (as it should): it is the flux of $\p U_2\cap\dots\cap\p U_d$ through $U_1$. See \Cref{fig:X}. 
\end{itemize}

\subsection{Main result}
Our main result is an elementary proof of the higher-dimensional analogue of the $d=2$ formula 
\begin{equation}
    c_{U_1,U_2}(P) = \X_{U_1,U_2} \cdot c_{H_1,H_2}(P),
\end{equation}
proved in \cite[Thm.~4]{DZ24}. It recovers a recent result of Martinian--Stoiber \cite{MS26}.

\begin{theorem}\label{thm:main}
If $U$ is a transverse family of subsets of $\R^d$, then
\begin{equation}\label{eq:main}
    c_U(P) = \X_U\cdot c_H(P).
\end{equation}
\end{theorem}

\begin{figure}[t]
\centering
\begin{minipage}{0.45\linewidth}
\centering
\includegraphics{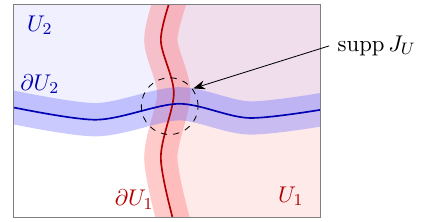}
\end{minipage}\hfill
\begin{minipage}{0.5\linewidth}
\setlength{\captionindent}{0pt}
\caption{In 2D, $f_1$ and $f_2$ are constant outside the $1$-neighborhoods of $\p U_1$ and $\p U_2$ (shaded bands). Hence $J_U = \det(\nabla f_1,\nabla f_2)$ vanishes outside the intersection of the bands (dashed), which is bounded by transversality.}
\label{fig:support}
\end{minipage}
\end{figure}

\begin{figure}[b]
\centering
\includegraphics[page=1]{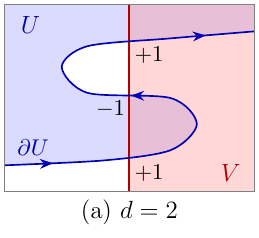}%
\hspace{0.25\linewidth}%
\includegraphics[page=2]{fig_intersection}
\caption{(a) In 2D, $\X$ counts the signed entrances of $\p U$ (oriented with $U$ on its left) into $V$, as in \cite[Fig.~1]{DZ24}. (b) In dimension $2n$, the hypersurfaces $\p U_2,\dots,\p U_{2n}$ meet along a curve $\gamma = \p U_2\cap\dots\cap\p U_{2n}$, and $\X_U$ counts the signed entrances of $\gamma$ into $U_1$ (schematic: $\p U_1$ is a hypersurface). In both pictures, $\X = +1-1+1 = 1$.}
\label{fig:X}
\end{figure}

The beauty of the formula \eqref{eq:main} is that it factors the physical quantity $c_U(P)$ into an integer, $\X_U$, depending only on the geometry, and an intrinsic quantity $c_H(P)$, depending only on the intrinsic nature of the topological phase. Ludewig and Thiang proved a geometric bulk-edge correspondence and along the way showed that, for idempotents with finite propagation, a universal multiple of $c_U(P)$ is an integer, which depends on $U$ only through its coarse cohomology class \cite{LT23,LT26}. In \cite{DZ24}, we  made this dependence explicit in the physical case $d=2$, proving \eqref{eq:DZ}. A consequence -- which could not directly follow from \cite{LT23,LT26} -- is that the absolutely continuous spectrum must fill the bulk spectral gap when truncating a non-trivial topological insulator to a region containing a parabola \cite{DSZ}. In a recent paper, Martinian--Stoiber  \cite{MS26} proved a geometric bulk-edge correspondence in all dimensions, and obtained \eqref{eq:main}.

The proof of \cite{MS26} uses sophisticated arguments from K-theory. It goes through index pairings on the uniform Roe algebra and KK-theory. The novelty of the present paper is to present a self-contained, elementary proof, accessible to mathematicians with only a basic background in differential geometry. It is almost entirely due to Claude Opus 5.5. It relies only on an algebraic identity satisfied by the Kubo formula (due to Kitaev \cite{Kit06}), relating it to simplicial forms; Stokes' theorem; and the Poincar\'e lemma.

\subsection{Proof strategy} The proof of \cite[Thm.~4]{DZ24} (which applies only to dimension $2$) deforms the sets $U_1, U_2$ and relies on planar geometry and topology. A key fact is that two points on a circle can be joined by a line. To adapt this approach to higher dimensions, one would need the -- obviously wrong -- analogous statement: a hypersurface in the sphere is the intersection of a sphere and a plane. Relying on \cite{DZ24} to prove
\Cref{thm:main} is then hopeless.

Instead of deforming the sets $U_j$ as in \cite{DZ24}, we take a more flexible approach and deform a more general object: a cochain. It turns out that such deformations do not impact $c_U(P)$. Specifically, we write $c_U(P)$ as a pairing:
\begin{equation}\label{eq:pairing-intro}
    c_U(P) = \sum_{\bx\in(\Z^d)^{d+1}} \Pi_P(\bx) \cdot \varphi_U(\bx),
\end{equation}
where $\Pi_P(\bx)$ depends on $P$ only, and $\varphi_U(\bx)$ depends on $U$ only. See \Cref{lem:cochain}. We call $\varphi_U$ a cochain because it takes the form of an integral over (suitably oriented) convex sets
\begin{equation}\label{eq:phiU-int}
\begin{gathered}
    \varphi_U(\bx) \defeq \int_{[\Gg_U(\bx)]} \w = I[\w]\big(\Gg_U(\bx)\big), \qquad     I[\alpha](\by) \defeq \int_{[\by]} \alpha, 
    \\ 
    \Gg_U(\bx) \defeq \big(G_U(x_0),\dots,G_U(x_d)\big), \qquad G_U(x) \defeq \big(\1_{U_1}(x),\dots,\1_{U_d}(x)\big).
\end{gathered}
\end{equation}
In \eqref{eq:phiU-int}, $\w$ is the Euclidean volume form and $[\by]$ denotes the oriented simplex  $\conv \by$, $\by = (y_0,\dots,y_d) \in (\R^d)^{d+1}$. Expressions of bulk topological markers as cochains go back to Kitaev \cite[App.~C]{Kit06}; for related mathematical identities without explicit connections to condensed matter physics, see \cite{Con85,CM90,Roe93}.
This realization has since been developed in several directions: see for instance real-space formulas for higher Chern numbers \cite{SHF25}, quantization via coarse cohomology \cite{LT23, LT26, MS26}, and interacting gapped systems \cite{KS, KSp20}.

The key step in the proof of \Cref{thm:main} is a homotopy argument, \Cref{lem:prism}. In differential-geometry terms, it shows that for each $\bx$, there exists a decaying $(d-1)$-form $\beta_\bx$ such that
\begin{equation}\label{eq:phiU-homotopy}
    \int_{[\Gg_U(\bx)]} \w = \int_{[\bx]} F_U^*\w + \int_{[\bx]} d\beta_\bx.
\end{equation}
Equivalently, using the notation $I$,
\begin{equation}
    I[\w]\big(\Gg_U(\bx)\big) = I[F_U^*\w](\bx) + I[d\beta_\bx](\bx).
\end{equation}
Using Stokes' theorem, we justify that $\int_{[\bx]} d\beta_\bx$ is a coboundary -- see \eqref{eq:coboundary} for the definition. For what matters to us at this point, this is an algebraic property that implies a vanishing pairing with $\Pi_P$:
\begin{equation}
    \sum_{\bx\in(\Z^d)^{d+1}} \Pi_P(\bx) \cdot I[d\beta_\bx](\bx) = 0,
\end{equation}
see \Cref{lem:pairing}.
Therefore, from $F_U^* \w = J_U \w$, we deduce that
\begin{equation}
    c_U(P) = \sum_{\bx\in(\Z^d)^{d+1}} \Pi_P(\bx) \cdot I[\w]\big(\Gg_U(\bx)\big) = \sum_{\bx\in(\Z^d)^{d+1}} \Pi_P(\bx) \cdot I[J_U \w](\bx).
\end{equation}

From Poincar\'e's lemma, $J_U \w$ and $\X_U \cdot J_H \w$ are cohomologically equivalent (equal up to the differential of a compactly supported $(d-1)$-form). By the above argument, this exact volume form does not contribute to $c_U(P)$. It follows that 
\begin{equation}
    c_U(P) = \sum_{\bx\in(\Z^d)^{d+1}} \Pi_P(\bx) \cdot I[J_U \w](\bx) = \X_U \cdot \sum_{\bx\in(\Z^d)^{d+1}} \Pi_P(\bx) \cdot I[J_H \w](\bx) = \X_U \cdot c_H(P).
\end{equation}
This is \eqref{eq:main}.

\subsection{Personal considerations} Let me expand on my relation with the present paper, since it relates deeply to the current developments around artificial intelligence. It provides an opportunity and a concrete basis for me to share some of my reflections.

In August 2024, Xiaowen Zhu and I posted \cite{DZ24}, where we proved the bulk-edge correspondence for topological insulators truncated along curved boundaries. This work established that if a $2$D topological phase with bulk index $c_H(P)$ is truncated to a set $U_1$, then its conductance into a transverse set $U_2$ is equal to
\begin{equation}\label{eq:DZ}
   c_{U_1,U_2}(P) = \X_{U_1,U_2} \cdot c_H(P).
\end{equation}
The formula \eqref{eq:DZ} is striking because it factors a physical quantity into one describing the geometry of the problem (encoded in $U_1, U_2$) and one describing functional-analytic aspects (encoded in $P$). Jacob Shapiro called \eqref{eq:DZ} a ``magical formula''.

The geometric Kubo formula $c_{U_1,U_2}(P)$ has the natural extension \eqref{eq:cU} to higher dimensions. The intersection number also naturally extends as \eqref{eq:1a}. I then formulated \eqref{eq:main} as a conjecture. In my view, \eqref{eq:DZ} and its multi-dimensional extension \eqref{eq:main} are elegant mathematical formulas that decorrelate intrinsic physical properties from geometry. I accept that its physical significance is limited in dimensions higher than $4$; yet I found beautiful mathematics even at the level of its sole formulation. 

Unfortunately, the approach of \cite{DZ24} could not be adapted to prove \Cref{thm:main}. At its core, \cite{DZ24} realizes a proper hypersurface in a large sphere as the intersection between this sphere and a (finite number of) hyperplane(s). This works in dimension $2$ (a proper curve and a sphere intersect at finitely many points), but this obviously fails in higher dimensions: a $(d-2)$-hypersurface in a $(d-1)$-sphere is almost never the intersection of this sphere with a hyperplane, i.e. a $(d-2)$-sphere! I knew that my approach with Xiaowen was too hands-on to prove \eqref{eq:main}, and that a more algebraic plan was needed. While such approaches for topological insulators do exist -- see e.g. \cite{Kit06,LT22,Lud23} for applications of cochains and \cite{BES94,Kit09,prodan-schulzbaldes,MS26} for K-theoretic arguments -- my background is in analysis and I did not have the time to learn something away from my roots. In January 2025, during a visit at the University of Washington, Guo Chuan Thiang proposed \eqref{eq:X} for the higher-dimensional intersection number. We however fell short of a concrete plan of action. In October 2025, and again in May 2026, I talked about the project with Guillaume Bal -- again the discussion remained there. Late August 2026, Martinian and Stoiber \cite{MS26} proved my conjecture, togther with a proof of the bulk-edge correspondence in all dimensions; I became aware of their work only at terminal stages of the present paper.  

My conjecture eventually emerged as particularly suited for LLMs: the central concepts were already developed in my paper with Xiaowen, it was clearly formulated, and I thought that a solution would involve tools that were present in the existing literature, but that I did not master. Early October 2026, I asked this question to Claude Opus 5.5. I fed it my paper with Xiaowen and a description of the project, together with an explanation of the above obstruction in higher dimensions. Claude proposed a 5-page solution within 40 minutes. Even in dimension two, Claude's proof is much simpler and more elegant than the one presented in \cite{DZ24}. The present paper is Claude's solution, together with a few of my own edits. Specifically, my contribution to the present paper has been to:
\begin{itemize}
    \item Identify the question and formulate the conjecture.
    \item Ask the question to Claude, with the material I found relevant.
    \item Simplify Claude's proof -- in particular I ``heroically'' removed an unnecessary part of its argument (Claude suggested the term heroic when I prompted it for a sarcastic acknowledgement, and I liked it).
    \item Expand technical details -- for instance, one has to carefully analyze the growth at infinity of the various cochains under consideration.
    \item Organize Claude's work and interpret the intuition behind it. 
    \item Track down (as well as possible, considering my limited background in algebraic tools) ownership of concepts and ideas.
\end{itemize}

On my end, Claude's biggest contribution has been to translate the algebraic tools of \cite{Roe93,CM90,Kit06,LT23,LT26} into a language that I could understand, and demonstrate their utility on a question that I cared about. While it is disturbing to see an LLM solve what could have been a full thesis in under an hour, I believe I learnt a lot from this project. I learnt about the recent result \cite{MS26} at terminal stages of this project, which Claude returned after I asked it to run a reference search. Looking at what Claude initially had access to, I believe Claude did not know about \cite{MS26} before (or at least did not base its response on it), but I could be wrong.

In my view, whether Claude knew about \cite{MS26} is not all that important. AI has demonstrated enough times that it is a powerful tool to prove deep mathematical results.  I prefer to stay out of controversies and prepare for a brutal transformation of mathematics. 
My thoughts go mainly to training future mathematicians. How should we teach the next generation of researchers, now that LLMs can solve Ph.D.-suited projects in just a few prompts? A few months ago, following this paper, I would have suggested extending \eqref{eq:main} to:
\begin{itemize}
    \item Other classes of topological phases? For instance, \eqref{eq:DZ} extends to time-reversal-invariant insulators in dimension $2$ \cite[Thms.~1 and 3]{DSZ26};
    \item Insulators with a mobility gap? See \cite{EGS05} for definitions.
    \item Insulating phases in the presence of interaction? See \cite{KS} for definitions.
\end{itemize}
In full honesty, these extensions are only of incremental importance: I do not expect them to yield deep new notions, techniques or connections. Yet they would have served as excellent training grounds for junior mathematicians, acting as stepping stones in their growth into independent, creative researchers. But today, I believe that an LLM, with adequate prompts, would produce satisfying responses. From the perspective of training graduate students, I see their significance dramatically reduced.  One can still learn from an AI-generated answer, like I did in this project; but I had enough time to think about the background, the challenges, and the elegance of the question beforehand.

For the time being, there remains a privileged space for the mathematicians who are able to identify the problems that shall guide us all: those that develop new concepts, new tools, new ideas. This is a visionary skill that is more difficult to acquire than mere technicity; it demands a broad knowledge and a refined taste, and there is no clear path to get there. This intuition may be what we would like to develop in the next generation of mathematicians: a good project for a graduate student may be to ``find your own'', and convince the mathematical community that it is right.

\subsection{Notation}
\begin{itemize}
    \item $\lr{x}\defeq(1+|x|^2)^{1/2}$ is the Japanese bracket.
    \item $P$ is a short-range projector on $\ell^2(\Z^d,\C^m)$, whose kernel $P(x,y)$ satisfies \eqref{eq:P-decay}.
    \item Given elements $A_1, \dots, A_d$ of an algebra, we set
    \begin{equation}
        A_1\wedge\cdots\wedge A_d \defeq \frac{1}{d!}\sum_{\sigma\in S_d}\sgn(\sigma)\,A_{\sigma(1)}\cdots A_{\sigma(d)}.
    \end{equation}
    \item $U=(U_1,\dots,U_d)$ is a transverse family of open subsets of $\R^d$, with transversality function $\Psi_U$, see \Cref{def:transverse}; $\1_{U_j}$ is the indicator function of $U_j$, also viewed as a multiplication operator. It has intersection number $\X_U$, defined in \eqref{eq:X}.
    \item $G_U$ is the vector of indicator functions $(\1_{U_1},\dots,\1_{U_d})$, and $\Gg_U$ applies $G_U$ to each entry of $\bx$, see \eqref{eq:phiU-int}.
    \item $F_U$ is the smoothed version $(\1_{U_1}*\eta,\dots,\1_{U_d}*\eta)$ of $G_U$, and $J_U = \det(\nabla F_U)$ is its Jacobian, see \Cref{def:X}.
    \item $\varphi_U$ is the $d$-cochain \eqref{eq:phiU}.
    \item $\bx=(x_0,\dots,x_k)$ is an element of $(\Z^d)^{k+1}$, $\conv \bx \subset \R^d$ is its convex hull, and $\diam\bx$ is the diameter of $\conv \bx$.
    \item $[\bx]$ is the oriented affine simplex with vertices $x_0,\dots,x_k$, see \S\ref{sec:proof};
    \item A $k$-cochain is a function on $(\Z^d)^{k+1}$; alternating and controlled cochains are defined in \Cref{def:cochain}.
    \item $\tau_P$ is the pairing with $\Pi_P$, both defined in \eqref{eq:tauP}.
    \item $\delta$ is the coboundary operator \eqref{eq:coboundary}.
     \item $I[\beta]$ is the cochain obtained by integrating the form $\beta$ over simplices, see \eqref{eq:phiU-int}.
     \item $\w$ is the volume form on $\R^d$.
\end{itemize}

\subsection*{Acknowledgments} I am grateful to Claude Opus 5.5, who provided an elegant solution to a project I cared deeply about (let's not cracher dans la soupe, comme on dit chez nous). I thank Guo Chuan Thiang for suggesting the formula \eqref{eq:X} for the intersection number. I also thank the NSF, which supports my research through CAREER grant DMS-2439949.

\section{Cochains}\label{sec:cochains}

In this section, we relate the geometric Kubo formula \eqref{eq:cU} to cochains, see \Cref{lem:cochain} below. This idea originally appeared in \cite[Appendix C]{Kit06}; see \cite{Con85,CM90,Roe93} for earlier occurrence in purely mathematical frameworks. The material below is standard for K-theoretic experts in condensed matter physics; we reproduce it here for the sake of completeness.

\begin{definition}\label{def:cochain}
A $k$-cochain is a function $\varphi : (\Z^d)^{k+1}\to\C$. We say that:
\begin{enumerate}[label=(\alph*)]
    \item $\varphi$ is alternating if
    \begin{equation}
        \varphi(x_{\pi(0)},\dots,x_{\pi(k)}) = \sgn(\pi)\,\varphi(x_0,\dots,x_k), \qquad \pi\in S_{k+1}.
    \end{equation}
    \item $\varphi$ is controlled if, for some $C,M\ge1$,
    \begin{equation}\label{eq:controlled}
        |\varphi(\bx)|\le C\lr{\diam \bx}^M\lr{x_0}^{-d-1}, \qquad \diam \bx \defeq \max_{i,j}|x_i-x_j|.
    \end{equation}
\end{enumerate}
\end{definition}

For $\bx\in(\Z^d)^{k+1}$ and a $k$-cochain $\varphi$, we set
\begin{equation}\label{eq:tauP}
    \Pi_P(\bx) \defeq \tr\big(P(x_0,x_1)P(x_1,x_2)\cdots P(x_k,x_0)\big), \qquad \tau_P(\varphi) \defeq \sum_{\bx\in(\Z^d)^{k+1}}\varphi(\bx) \cdot \Pi_P(\bx).
\end{equation}

\begin{lemma}\label{lem:convergence}
For every $N$ there is $C_N$ such that
\begin{equation}\label{eq:Pi-decay}
    |\Pi_P(\bx)| \le C_N\lr{\diam \bx}^{-N}\prod_{i=0}^{k-1}\lr{x_{i+1}-x_i}^{-N}.
\end{equation}
In particular, if $\varphi$ is controlled, the series \eqref{eq:tauP} defining $\tau_P(\varphi)$ converges absolutely.
\end{lemma}

\begin{proof} With $x_{k+1}\defeq x_0$, the decay \eqref{eq:P-decay} gives $|P(x_i,x_{i+1})|\le C_{2N}\lr{x_{i+1}-x_i}^{-2N}$ for each of the $k+1$ factors. Half the powers produce the product from $i=0$ to $k-1$ in \eqref{eq:Pi-decay}. The other half produce
\begin{equation}
    \prod_{i=0}^k\lr{x_{i+1}-x_i}^{-N} \le c_k^{-N} \lr{\diam \bx}^{-N}.
\end{equation}    
Indeed, we have     
\begin{equation}
    \prod_{i=0}^k\lr{x_{i+1}-x_i}\ \ge\ c_k\lr{\diam \bx},
\end{equation}
which holds because the closed path $x_0\to x_1\to\dots\to x_k\to x_0$ has length at least $\diam \bx$. This proves \eqref{eq:Pi-decay}.

Now let $\varphi$ be controlled. By \eqref{eq:controlled} and \eqref{eq:Pi-decay} with $N = M+d+1$,
\begin{equation}
    |\varphi(\bx) \cdot \Pi_P(\bx)| \le C\lr{x_0}^{-d-1}\prod_{i=0}^{k-1}\lr{x_{i+1}-x_i}^{-d-1}.
\end{equation}
Summing over $x_k$, then $x_{k-1}$, \dots, then $x_0$, we obtain:
\begin{equation}\label{eq:sum-bound}
    \sum_{\bx\in(\Z^d)^{k+1}}|\varphi(\bx) \cdot \Pi_P(\bx)| \le C\sum_{x_0\in\Z^d}\lr{x_0}^{-d-1}\Big(\sum_{y\in\Z^d}\lr{y}^{-d-1}\Big)^k < \infty.
\end{equation}
This completes the proof. 
\end{proof}

\begin{lemma}\label{lem:cochain}
Let $U$ be transverse and
\begin{equation}
    G_U \defeq (\1_{U_1},\dots,\1_{U_d}) : \Z^d\to\{0,1\}^d.
\end{equation}
Then $c_U(P) = \tau_P(\varphi_U)$, where $\varphi_U$ is the alternating and controlled $d$-cochain
\begin{equation}\label{eq:phiU}
    \varphi_U(x_0,\dots,x_d) \defeq \frac{1}{d!}\det\big[G_U(x_1)-G_U(x_0),\dots,G_U(x_d)-G_U(x_0)\big].
\end{equation}
\end{lemma}

For analogous statements, see \cite[(125)]{Kit06}, \cite[Lem.~2.16]{LT23} in dimension $2$ and \cite[Prop.~5.10, Lem.~5.11]{MS26}, \cite[Lem.~5.9]{LT26} in higher dimensions.

\begin{proof} The kernel of $[P,\1_{U_j}]$ is
\begin{equation}
    [P,\1_{U_j}](x,y) = P(x,y)\big(\1_{U_j}(y)-\1_{U_j}(x)\big).
\end{equation}
We label the summation variables so that the $k$-th commutator joins $x_{k-1}$ to $x_k$, and the first factor $P$ joins $x_d$ to $x_0$. By cyclicity of $\tr$, we get
\begin{equation}\label{eq:trace-expansion}
    \Tr\big(P[P,\1_{U_1}]\cdots[P,\1_{U_d}]\big) = \sum_{\bx\in(\Z^d)^{d+1}}\Pi_P(\bx) \cdot \prod_{k=1}^d\big(\1_{U_k}(x_k)-\1_{U_k}(x_{k-1})\big).
\end{equation}
This sum converges absolutely by \Cref{lem:convergence}, because the product
\begin{equation}\label{eq:AU}
    A_U(\bx) \defeq \prod_{k=1}^d\big(\1_{U_k}(x_k)-\1_{U_k}(x_{k-1})\big)
\end{equation}
emerging above defines a controlled $d$-cochain. Indeed, for $A_U(\bx)$ to be non-zero, one of the points $x_{k-1}, x_k$ needs to be in $U_k$, and the other in $U_k^c$. Hence
\begin{equation}
    \dist(x_0,\p U_k) \leq \max\big(|x_0-x_k|, |x_0 - x_{k-1}|\big) \leq \diam \bx.
\end{equation}
Summing over $k$, and using transversality, we obtain that if $A_U(\bx) \neq 0$ and $|x_0| \geq c^{-1}$,
\begin{equation}
    |x_0|^c \leq \Psi_U(x_0)\le 1+d\diam \bx.
\end{equation}
In particular, using that all factors in \eqref{eq:AU} are in $\{-1,0,1\}$, if $A_U(\bx)\neq 0$ and $|x_0| \geq c^{-1}$,
\begin{equation}\label{eq:AU-bound}
    |A_U(\bx)| \leq 1 \leq \big((1+d\diam \bx) |x_0|^{-c}\big)^{\tfrac{d+1}{c}} \leq C \lr{\diam \bx}^{\tfrac{d+1}{c}} \lr{x_0}^{-d-1}.
\end{equation}
Since $c\in(0,1)$, the exponent $\frac{d+1}{c}$ is larger than $1$, as required in \Cref{def:cochain}. We note that if $A_U(\bx) = 0$ or $|x_0| \leq c^{-1}$, then \eqref{eq:AU-bound} holds trivially (after potentially increasing $C$). Therefore $A_U$ is controlled.

Antisymmetrizing \eqref{eq:trace-expansion}, we obtain
\begin{equation}
\begin{aligned}
    c_U(P) &= \Tr\big(P\,[P,\1_{U_1}]\wedge\cdots\wedge[P,\1_{U_d}]\big) 
    \\
    &= \sum_{\bx\in(\Z^d)^{d+1}} \varphi_U(\bx) \cdot \Pi_P(\bx) = \tau_P(\varphi_U), 
 \end{aligned}
\end{equation}
where $\varphi_U$ is the $d$-cochain defined by
\begin{equation}
    \varphi_U(\bx) \defeq \frac{1}{d!}\det[v_1,\dots,v_d], \qquad v_k \defeq G_U(x_k)-G_U(x_{k-1}).
\end{equation}
Replacing each $v_k$ by $v_1+\dots+v_k = G_U(x_k)-G_U(x_0)$ does not change the above determinant, and confirms the identity \eqref{eq:phiU}. This identity immediately justifies that the cochain $\varphi_U$ is alternating. Moreover, it is a linear combination of $d!$ (controlled) cochains of the form \eqref{eq:AU} over permutations of $U = (U_1, \dots, U_d)$. So $\varphi_U$ is also controlled. This completes the proof.
\end{proof}

The coboundary of a $k$-cochain $\varphi$ is the $(k+1)$-cochain
    \begin{equation}\label{eq:coboundary}
        \delta\varphi(x_0,\dots,x_{k+1}) \defeq \sum_{i=0}^{k+1}(-1)^i\varphi(x_0,\dots,\widehat{x_i},\dots,x_{k+1}),
    \end{equation}
    where $\widehat{x_i}$ in \eqref{eq:coboundary} means that $x_i$ is omitted. If $\varphi$ is controlled, then so is $\delta\varphi$. Indeed, write $\bx=(x_0,\dots,x_{k+1})$. The terms with $i\ge1$ are bounded by $C\lr{\diam\bx}^M\lr{x_0}^{-d-1}$. For $i=0$, Peetre's inequality $\lr{x_0}\le\sqrt2\,\lr{x_0-x_1}\lr{x_1}$ gives
\begin{equation}\label{eq:peetre}
    \lr{x_1}^{-d-1}\le C\lr{\diam\bx}^{d+1}\lr{x_0}^{-d-1}.
\end{equation}
Summing over $i$ provides \eqref{eq:controlled} for $\delta \varphi$, so $\delta\varphi$ is controlled as claimed.

\begin{lemma}\label{lem:pairing}
If $\psi$ is an alternating controlled $(d-1)$-cochain, then $\tau_P(\delta\psi) = \tau_P(\psi) = 0$.
\end{lemma}

This result already appeared in \cite[II, Prop.~14]{Con85} and \cite[\S4.2]{Roe93}; and when $d=2$, in \cite[(123)]{Kit06}. See also \cite[(48)--(50)]{KS} and \cite[(16), (23)]{KSp20} for statements related to interacting systems.

\begin{proof} By definition,
\begin{equation}
    \tau_P(\delta\psi) = \sum_{i=0}^d(-1)^i\sum_{\bx\in(\Z^d)^{d+1}}\psi(x_0,\dots,\widehat{x_i},\dots,x_d) \cdot \Pi_P(\bx).
\end{equation}
The infinite sum above is absolutely convergent. Indeed, set $N = M+d+1$. If $i \geq 1$, we have:
\begin{align}
    \sum_{\bx\in(\Z^d)^{d+1}}\big|\psi(x_0,\dots,\widehat{x_i},\dots,x_d) \cdot \Pi_P(\bx)\big| &\leq C\sum_{\bx\in(\Z^d)^{d+1}} \lr{x_0}^{-d-1} \lr{\diam \bx}^{M-N} \prod_{j=0}^{d-1}\lr{x_{j+1}-x_j}^{-N} \notag\\
    &\leq C\sum_{\bx\in(\Z^d)^{d+1}} \lr{x_0}^{-d-1} \prod_{j=0}^{d-1}\lr{x_{j+1}-x_j}^{-N}, \label{eq:face-i}
\end{align}
and if $i=0$, then by \eqref{eq:peetre}:
\begin{align}
    \sum_{\bx\in(\Z^d)^{d+1}}\big|\psi(\widehat{x_0},x_1,\dots,x_d) \cdot \Pi_P(\bx)\big| &\leq C\sum_{\bx\in(\Z^d)^{d+1}} \lr{x_1}^{-d-1} \lr{\diam \bx}^{M-N} \prod_{j=0}^{d-1}\lr{x_{j+1}-x_j}^{-N} \notag\\
    &\leq C\sum_{\bx\in(\Z^d)^{d+1}} \lr{x_0}^{-d-1} \lr{\diam \bx}^{M+d+1-N} \prod_{j=0}^{d-1}\lr{x_{j+1}-x_j}^{-N} \notag\\
    &\leq C\sum_{\bx\in(\Z^d)^{d+1}} \lr{x_0}^{-d-1} \prod_{j=0}^{d-1}\lr{x_{j+1}-x_j}^{-N}. \label{eq:face-0}
\end{align}
The right-hand sides of \eqref{eq:face-i} and \eqref{eq:face-0} converge as in \eqref{eq:sum-bound}.

In the $i$-th term, we sum over $x_i$ first. Since $P^2=P$,
\begin{equation}
    \sum_{x_i\in\Z^d}P(x_{i-1},x_i)P(x_i,x_{i+1}) = P(x_{i-1},x_{i+1}) \qquad (\text{indices mod } d+1),
\end{equation}
so the $i$-th term equals $(-1)^i\tau_P(\psi)$, and, since $d$ is even, $\tau_P(\delta\psi) = \tau_P(\psi)$. The permutation
\begin{equation}
    (y_0,\dots,y_{d-1})\mapsto(y_1,\dots,y_{d-1},y_0)
\end{equation}
leaves $\Pi_P$ invariant and flips the sign of $\psi$, because a $d$-cycle has signature $(-1)^{d+1} = -1$. Hence $\tau_P(\psi) = 0$, and this completes the proof.
\end{proof}

\section{Proof of \texorpdfstring{\Cref{thm:main}}{the theorem}}\label{sec:proof}

For $\bx\in(\R^d)^{k+1}$, write $[\bx]$ for the oriented affine simplex $\sigma_{\bx}:\Delta^k\to\R^d$ with vertices $x_0,\dots,x_k$, where $\Delta^k$ is the standard $k$-simplex:
\begin{equation}
    \Delta^k \defeq \big\{t\in[0,1]^{k+1} : t_0+\dots+t_k=1\big\}, \qquad \sigma_{\bx}(t) \defeq \sum_{i=0}^k t_ix_i.
\end{equation}
For a $k$-form $\beta$, we set
\begin{equation}
    I[\beta](\bx)\defeq\int_{[\bx]}\beta, \qquad\text{so that Stokes' formula reads}\qquad I[d\beta] = \delta I[\beta].
\end{equation}

In this section, we use de Rham cohomology to prove \Cref{thm:main}. The key element is to interpret the formula \eqref{eq:phiU} as simplicial integration: we write  it as
\begin{equation}
    \varphi_U(\bx) = \int_{[\Gg_U(\bx)]} \w = \big(\Gg_U^* I[\w]\big)(\bx), \qquad \Gg_U(\bx)\defeq\big(G_U(x_0),\dots,G_U(x_d)\big),
\end{equation}
where $\w$ is the standard volume form on $\R^d$. It turns out that $\Gg_U^* I[\w]$ is not $I[G_U^* \w]$. This is false even at the formal level, setting smoothness issues aside ($G_U^* \w$ is, technically speaking, not even defined): $\Ff^*I[\w] = I[F^* \w]$ generally fails for non-affine $F$. What holds instead is an identity up to a coboundary:

\begin{lemma}\label{lem:prism}
Let $F:\R^d\to\R^d$ be smooth and $G:\Z^d\to\R^d$ be any map. Then there exists an alternating $(d-1)$-cochain $q$ such that
\begin{equation}\label{eq:prism}
    \Gg^*I[\w] = I[F^*\w]+\delta q, \qquad \Gg(\bx) \defeq \big(G(x_0),\dots,G(x_d)\big).
\end{equation}
\end{lemma}

\begin{proof} As for $\Gg_U$, we use the notation
\begin{equation}
    \Gg(\by)\defeq\big(G(y_0),\dots,G(y_k)\big), \qquad \big(\Gg^*I[\w]\big)(\bx)\defeq I[\w]\big(\Gg(\bx)\big),
\end{equation}
for $\by\in(\Z^d)^{k+1}$ and $\bx\in(\Z^d)^{d+1}$. Let $H_{\by}:\Delta^k\times[0,1]\to\R^d$ be the straight-line homotopy from $F\circ\sigma_{\by}$ to $\sigma_{\Gg(\by)}$, and $q$ the $(d-1)$-cochain:
\begin{equation}
    H_{\by}(t,s)\defeq(1-s)\,F\big(\sigma_{\by}(t)\big)+s\,\sigma_{\Gg(\by)}(t), \qquad q(\by)\defeq-\int_{\Delta^{d-1}\times[0,1]}H_{\by}^*\w.
\end{equation}

The cochain $q$ is alternating. Indeed, let $\pi$ be a permutation of $\{0,\dots,d-1\}$, and let $A_\pi:\Delta^{d-1}\to\Delta^{d-1}$ be the affine map permuting the vertices of $\Delta^{d-1}$ according to $\pi$, i.e.\ $(A_\pi t)_{\pi(i)}=t_i$; its orientation is $\sgn(\pi)$, as one checks on transpositions. With
\begin{equation}
    \by_\pi\defeq\big(y_{\pi(0)},\dots,y_{\pi(d-1)}\big), \qquad\text{we have}\qquad \sigma_{\by_\pi}=\sigma_{\by}\circ A_\pi, \qquad \sigma_{\Gg(\by_\pi)}=\sigma_{\Gg(\by)}\circ A_\pi.
\end{equation}
Hence $H_{\by_\pi}=H_{\by}\circ(A_\pi\times\mathrm{id})$, and the change of variables formula gives $q(\by_\pi)=\sgn(\pi)\,q(\by)$.

We now prove the identity \eqref{eq:prism}. Fix $\bx\in(\Z^d)^{d+1}$ and set $\bx^{(i)}\defeq(x_0,\dots,\widehat{x_i},\dots,x_d)$. We have $d\w=0$; hence $d(H_{\bx}^*\w)=H_{\bx}^*(d\w)=0$. Moreover, $H_{\bx}$ restricts to $H_{\bx^{(i)}}$ on $\p_i\Delta^d\times[0,1]$, where $\p_i\Delta^d\cong\Delta^{d-1}$ is the $i$-th face of $\Delta^d$. Since $d$ is even, the oriented boundary of $\Delta^d\times[0,1]$ is
\begin{equation}
    \p\big(\Delta^d\times[0,1]\big) = \Delta^d\times\{1\} - \Delta^d\times\{0\} + \sum_{i=0}^d(-1)^i\,\p_i\Delta^d\times[0,1].
\end{equation}
Therefore, Stokes' formula gives
\begin{equation}
\begin{aligned}
    0 &= \int_{\Delta^d\times[0,1]} d\big(H_{\bx}^*\w\big) = \int_{\p(\Delta^d\times[0,1])} H_{\bx}^*\w \\
    &= \int_{\Delta^d}\sigma_{\Gg(\bx)}^*\w - \int_{\Delta^d}(F\circ\sigma_{\bx})^*\w + \sum_{i=0}^d(-1)^i\int_{\Delta^{d-1}\times[0,1]}H_{\bx^{(i)}}^*\w \\
    &= \big(\Gg^*I[\w]\big)(\bx) - I[F^*\w](\bx) - \delta q(\bx).
\end{aligned}
\end{equation}
This completes the proof.
\end{proof}

We next apply \Cref{lem:prism} to the case $F=F_U$ and $G=G_U$, with the goal of plugging the resulting formula into the identity $c_U(P)=\tau_P(\varphi_U)$ of \Cref{lem:cochain} and cancelling the coboundary term. To do this, we need to show that the coboundary is controlled:

\begin{proposition}\label{prop:smoothing}
The $(d-1)$-cochain $Q_U$ constructed by \Cref{lem:prism} with $F=F_U$ and $G=G_U$ is controlled. In particular,
\begin{equation}\label{eq:cU-smooth}
    c_U(P) = \tau_P\big(I[J_U\,\w]\big).
\end{equation}
\end{proposition}

\begin{proof} Since $\varphi_U=\Gg_U^*I[\w]$ and $F_U^*\w=\det(DF_U)\,\w=J_U\,\w$, \Cref{lem:prism} with $F=F_U$ and $G=G_U$ gives
\begin{equation}\label{eq:smoothing}
    \varphi_U = I[J_U\,\w] + \delta Q_U.
\end{equation}
Recall from the proof of \Cref{lem:prism} that, for $\by\in(\Z^d)^d$,
\begin{equation}
    Q_U(\by)=-\int_{\Delta^{d-1}\times[0,1]}H_{\by}^*\w, \qquad H_{\by}(t,s)=(1-s)\,F_U\big(\sigma_{\by}(t)\big)+s\,\sigma_{\Gg_U(\by)}(t).
\end{equation}

We now show that $Q_U$ is controlled. First, since $H_{\by}$ is Lipschitz with constant at most $C\lr{\diam\by}$, we have $|Q_U(\by)|\le C\lr{\diam\by}^d$. Second, assume that $\dist(\conv\by,\p U_j)\ge1$ for some $j$. The $1$-neighborhood of $\conv\by$ is connected and does not meet $\p U_j$, so $\1_{U_j}$ is constant on it, equal to $\1_{U_j}(y_0)$. Since $\eta$ is supported in $B(0,1)$, it follows that $f_j=\1_{U_j}(y_0)$ on $\conv\by$; see \Cref{fig:smoothing}. Since moreover $(G_U)_j(y_i)=\1_{U_j}(y_i)=\1_{U_j}(y_0)$ and the $t_i$ sum to one, the $j$-th component of $H_{\by}$ is constant:
\begin{equation}
    (H_{\by})_j(t,s) = (1-s)\,\1_{U_j}(y_0) + s\sum_{i=0}^{d-1} t_i\,\1_{U_j}(y_0) = \1_{U_j}(y_0).
\end{equation}
Hence $H_{\by}^*\w=0$ and $Q_U(\by)=0$.

Therefore, if $Q_U(\by)\neq0$, then for each $j$ there exist $p\in\conv\by$ and $z\in\p U_j$ with $|p-z|<1$, so that
\begin{equation}
    \dist(y_0,\p U_j) \leq |y_0-p| + |p-z| \leq \diam\by + 1.
\end{equation}
Summing over $j$ gives $\Psi_U(y_0)\le 1+d(1+\diam\by)$. If moreover $|y_0|\ge c^{-1}$, then transversality yields $|y_0|^c\le C\lr{\diam\by}$, and therefore
\begin{equation}
    |Q_U(\by)| \le C\lr{\diam\by}^d \le C\lr{\diam\by}^{d+\tfrac{d+1}{c}}\lr{y_0}^{-d-1}.
\end{equation}
This bound also holds, after increasing $C$, when $Q_U(\by)=0$ or $|y_0|<c^{-1}$. Hence $Q_U$ is controlled.

\begin{figure}[t]
\centering
\begin{minipage}{0.35\linewidth}
\centering
\includegraphics{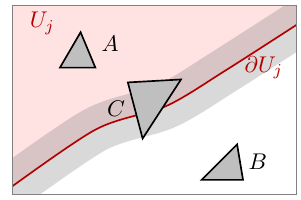}
\end{minipage}\hfill
\begin{minipage}{0.6\linewidth}
\setlength{\captionindent}{0pt}
\caption{The function $f_j=\1_{U_j}*\eta$ coincides with $\1_{U_j}$ outside the $1$-neighborhood of $\p U_j$ (shaded band). On the simplices $A$ and $B$, which avoid it, $f_j$ is constant; hence $H_{\by}^*\w=0$ and $Q_U(\by)=0$. Only simplices such as $C$ can contribute.}
\label{fig:smoothing}
\end{minipage}
\end{figure}

Finally, $I[F_U^*\w]= I[J_U\,\w] = \varphi_U-\delta Q_U$ is controlled, as the difference of two controlled cochains. Therefore \Cref{lem:cochain} and \Cref{lem:pairing} give
\begin{equation}
    c_U(P) = \tau_P(\varphi_U) = \tau_P\big(I[F_U^*\w]\big)+\tau_P(\delta Q_U) = \tau_P\big(I[J_U\,\w]\big).
\end{equation}
This completes the proof. 
\end{proof}

We are now ready for the proof of \Cref{thm:main}.

\begin{proof}[Proof of \Cref{thm:main}] Let $\rho\defeq J_U-\X_U J_H$. By the remarks following \Cref{def:X}, $\rho\in C_0^\infty(\R^d)$, and $\int_{\R^d}\rho = 0$ because $\X_H=1$. By the Poincar\'e lemma for compact supports \cite[Cor.~4.7.1, p.~39]{BT82}, there exists a compactly supported $(d-1)$-form $\beta$ such that
\begin{equation}
    \rho\,\w = d\beta.
\end{equation}
Stokes' formula gives $I[\rho\,\w] = I[d\beta] = \delta I[\beta]$. By an argument similar to that presented in the proof of \Cref{lem:prism}, the $(d-1)$-cochain $I[\beta]$ is alternating. Because the range $\conv\bx$ of a $(d-1)$-dimensional simplex $[\bx]$ has area controlled by $\diam \bx$, and $\beta$ is smooth and compactly supported,
\begin{equation}
    |I[\beta](\bx)| \le  C\lr{\diam\bx}^{d-1}.
\end{equation}

Let $R$ be such that $\supp\beta \subset B(0,R)$. If $I[\beta](\bx)\neq0$, then $\conv\bx$ meets $\supp\beta$: there exists $y\in\conv\bx$ with $|y|\le R$. Since $|y-x_0|\le\diam\bx$, we get $|x_0|\le R+\diam\bx$, hence $\lr{x_0}\le C\lr{\diam\bx}$ and
\begin{equation}
    |I[\beta](\bx)| \le C\lr{\diam\bx}^{d-1} \le C\lr{\diam\bx}^{2d}\lr{x_0}^{-d-1}.
\end{equation}
This bound holds trivially when $I[\beta](\bx)=0$. Hence $I[\beta]$ is controlled. 

Therefore, by \Cref{lem:pairing}:
\begin{equation}
   \tau_P\big(I[\rho\,\w]\big) = \tau_P\big(\delta I[\beta]\big) = 0.
\end{equation}
On the other hand, from $\rho = J_U-\X_U J_H$ and \Cref{prop:smoothing}, we also have that:
\begin{equation}
    \tau_P\big(I[\rho\,\w]\big) = \tau_P\big(I[J_U\,\w]\big) - \X_U\cdot \tau_P\big(I[J_H\,\w]\big) = c_U(P) - \X_U\cdot c_H(P).
\end{equation}
This completes the proof. 
\end{proof}

\end{document}